\documentclass[a4paper, 11pt]{article}
\usepackage{amssymb,amsmath,amsthm,multirow,color,xcolor,cancel,bbm,tikz}

\usepackage[margin=2cm]{geometry}

\newtheorem{proposition}{Proposition\setcounter{claimcounter}{0}}
\newtheorem{lemma}{Lemma\setcounter{claimcounter}{0}}
\newtheorem{theorem}{Theorem\setcounter{claimcounter}{0}}
\newtheorem{corollary}{Corollary}

\newcounter{claimcounter}
\newtheorem*{claim*}{{\it Claim}}

\newenvironment{subproof}{\begin{proof}[Proof of Claim]}{\end{proof}}

\newcommand{\rem}{\mathrm{rem}}
\newcommand{\cov}{\mathrm{cr}}
\newcommand{\fix}{\mathrm{fix}}
\newcommand{\linear}{\mathrm{lin}}
\newcommand{\affine}{\mathrm{aff}}
\newcommand{\G}{\mathrm{G}}
\newcommand{\GF}{\mathrm{GF}}
\newcommand{\inn}{\mathrm{inn}}
\newcommand{\ind}{\mathrm{ind}}
\newcommand{\Ima}{\mathrm{Im}}

\newcommand{\dH}{d_\mathrm{H}}
\newcommand{\wH}{w_\mathrm{H}}
\newcommand{\VH}{V_\mathrm{H}}

\title{Finite dynamical systems, hat games, and coding theory}

\author{Maximilien Gadouleau\footnote{School of Engineering and Computing Sciences, Durham University, Durham, UK. Email: \texttt{m.r.gadouleau@durham.ac.uk}\newline Work supported by EPSRC grant EP/K033956/1} }

\begin{document}

\maketitle

\begin{abstract}
The dynamical properties of finite dynamical systems (FDSs) have been investigated in the context of coding theoretic problems, such as network coding and index coding, and in the context of hat games, such as the guessing game and Winkler's hat game. In this paper, we relate the problems mentioned above to properties of FDSs, including the number of fixed points, their stability, and their instability. We first introduce the guessing dimension and the coset dimension of an FDS and their counterparts for digraphs. Based on the coset dimension, we then strengthen the existing equivalences between network coding and index coding. We also introduce the instability of FDSs and we study the stability and the instability of digraphs. We prove that the instability always reaches the size of a minimum feedback vertex set. We also obtain some non-stable bounds independent of the number of vertices of the graph. We then relate the stability and the instability to the guessing number. We also exhibit a class of sparse graphs with large girth that have high stability and high instability; our approach is code-theoretic and uses the guessing dimension. Finally, we prove that the affine instability is always asymptotically greater than or equal to the linear guessing number.
\end{abstract}

\section{Introduction} \label{sec:intro}

\textbf{Finite dynamical systems} (FDSs) have been used to represent a network of interacting entities as follows. A network of $n$ entities has a state $x= (x_1,\dots, x_n) \in A^n$, represented by a variable $x_v$ taking its value in a finite alphabet $A$ on each entity $v$. The state then evolves according to a deterministic function $f = (f_1,\dots,f_n) : A^n \to A^n$, where $f_v : A^n \to A$ represents the update of the local state $x_v$. Although different update schedules have been studied, we are focusing on the parallel update schedule, where all entities update their state at the same time, and $x$ becomes $f(x)$. FDSs have been used to model gene networks \cite{Kau69, Tho73, TK01, Jon02, KS08}, neural networks \cite{MP43, Hop82, Gol85}, social interactions \cite{PS83, GT83} and more (see \cite{TD90, GM90}).

The structure of an FDS $f: A^n \to A^n$ can be represented via its \textbf{interaction graph} $G(f)$, which indicates which update functions depend on which variables. More formally, $G(f)$ has $\{1,\dots,n\}$ as vertex set and there is an arc from $u$ to $v$ if $f_v(x)$ depends on $x_u$. In different contexts, the interaction graph is known--or at least well approximated--, while the actual update functions are not. One main problem of research on FDSs is then to predict their dynamics according to their interaction graphs.\\
~\\
Network coding is a technique to transmit information through networks, which can significantly improve upon routing in theory \cite{ACLY00, YLCZ06}. At each intermediate node $v$, the received messages $x_{u_1}, \ldots, x_{u_k}$ are combined, and the combined message $f_v(x_{u_1},\ldots,x_{u_k})$ is then forwarded towards its destinations. The main problem is to determine which functions $f_v$ can transmit the most information. In particular, the \textbf{network coding solvability problem} tries to determine whether a certain network situation, with a given set of sources, destinations, and messages, is solvable, i.e. whether all messages can be transmitted to their destinations. This problem being very difficult, different techniques have been used to tackle it, including matroids \cite{DFZ06}, Shannon and non-Shannon inequalities for the entropy function \cite{DFZ07, Rii07}, error-correcting codes \cite{GR11}, and closure operators \cite{Gad13, Gad14}. 

The network coding solvability problem can be recast in terms of fixed points of FDSs as follows \cite{Rii07, Rii07a}. The so-called \textbf{guessing number} \cite{Rii07} of a digraph $D$ is the logarithm of the maximum number of fixed points over all FDSs $f$ whose interaction graph is a subgraph of $D$: $G(f) \subseteq D$. The guessing number is always upper bounded by the size of a minimum feedback vertex set of $D$; if equality holds, we say that $D$ is \textbf{solvable} and the FDS $f$ reaching this bound is called a solution. Then, a network coding instance $N$ is solvable if and only if some digraph $D_N$ closely related to the instance $N$ is solvable.

\textbf{Linear network coding} is the most popular kind of network coding, where the intermediate nodes can only perform linear combinations of the packets they receive \cite{LYC03}. The network coding instance $N$ is then linearly solvable if and only if $D_N$ admits a linear solution. Many interesting classes of linearly solvable digraphs have been given in the literature (see \cite{Rii06, GR11}). On the other hand, graphs which are not linearly solvable have been exhibited in \cite{Rii04, Rii07b, GRF15}. 

\textbf{Index coding} is a new way to broadcast information to different receivers who have different partial information. The problem of index coding is to find the smallest index code, i.e. the minimum amount of information to transmit to all destinations so that they all can gather the same amount of information. Network coding solvability is closely related to index coding \cite{ESG10, BBJK11}. In particular, the length of a minimal index code (for a given digraph) is the same as the \textbf{information defect} \cite{Rii07, GR11}. Since a graph is solvable if and only if it is solvable in the sense of information defect \cite{Rii07, Gad13}, there is an equivalence between index coding and network coding. This equivalence was independently given in \cite{EEL13} and extended to storage capacity in \cite{Maz14}. In fact, there are two additional equivalences: an asymptotic version is given in \cite{GR11}, and a version in the linear case can be found in \cite{GR11}.\\
~\\
\textbf{Hat games} are an increasingly popular topic in combinatorics. Typically, a hat game involves $n$ players, each wearing a hat that can take a colour from a given set of $q$ colours. No player can see their own hat, but each player can see some subset of the other hats. All players are asked to guess the colour of their own hat at the same time. For an extensive review of different hat games, see \cite{Krz12}. Different variations have been proposed: for instance, the players can be allowed to pass \cite{Ebe98}, or the players can guess their respective hat's colour sequentially \cite{Krz10}. The variation in \cite{Ebe98} mentioned above has been investigated further (see \cite{Krz12}) for it is strongly connected to coding theory via the concept of covering codes \cite{CHLL97}; in particular, some optimal solutions for that variation involve the well-known Hamming codes \cite{EMV03}. 

In the variation called the ``\textbf{guessing game},'' players are not allowed to pass, and must guess simultaneously \cite{Rii07}. The team wins if everyone has guessed their colour correctly; the aim is to maximise the number of hat assignments which are correctly guessed by all players. This version of the hat game then aims to determine the guessing number of a digraph. 

In \textbf{Winkler's hat game}, the players are not allowed to pass, and the team scores as many points as players guessing correctly. The aim is then to construct a guessing function $f$ which guarantees a score for any possible configuration of hats \cite{Win01}. The relation between Winkler's hat game and auctions has been revealed in \cite{AFGHIS11} and developed in \cite{BNW15}.\\
~\\
In this paper, we relate the problems of network coding, index coding, the guessing game, and Winkler's hat game, to the (in)stability of finite dynamical systems. We first review the relevant background in Section \ref{sec:background}. In Section \ref{sec:guessing_code}, we introduce the guessing dimension and the coset dimension of an FDS and their counterparts for digraphs. Based on the coset dimension, we then strengthen the known equivalences between network coding and index coding in Theorem \ref{th:g+l=n}. In Section \ref{sec:instability}, we introduce the instability of FDSs and we study the stability and the instability of digraphs. We prove that the instability always reaches the size of a minimum feedback vertex set in Theorem \ref{th:i=tau}. We also give obtain some non-stable bounds based only on the size of a minimum feedback vertex set (and independent on the number of vertices of the graph) in Theorem \ref{th:stability_upper_bound}. We then relate the stability and the instability to the guessing number in Theorem \ref{th:i_v_g}. Finally, we obtain  results for linear and affine FDSs in Section \ref{sec:linear}. In Theorem \ref{th:high_i}, we exhibit a class of sparse graphs with high girth and high affine (in)stability; our approach is code-theoretic and uses the guessing dimension. We also prove that the affine instability is always asymptotically greater than or equal to the linear guessing number in Theorem \ref{th:i_affine}.

\section{Background} \label{sec:background}

\subsection{Notation for digraphs and finite dynamical systems} \label{sec:notation}

Let $n$ be a positive integer, $V = \{1, \dots, n\}$, and $D$ be a loopless digraph on $V$, i.e. $D = (V,E)$ with $E \subseteq V^2 \setminus\{(v,v) : v \in V\}$. Paths and cycles are always directed. The \textbf{girth} of $D$ is the minimum length of a cycle in $D$ The maximum number of vertex disjoint cycles in $D$ is denoted $\nu(D)$. A \textbf{feedback vertex set} is a set of vertices $I$ such that $D-I$ has no cycles. The minimum size of a feedback vertex set is denoted $\tau(D)$. If $J \subseteq V$, then $D[J]$ is the subgraph of $D$ induced by $J$. If this graph is acyclic, then we can sort $J$ in \textbf{topological order}: $J = \{j_1, \dots, j_k\}$ where $(j_a, j_b) \in D$ only if $a < b$. The in-neighbourhood of a vertex $v$ in $D$ is denoted by $\inn(v)$, its in-degree is denoted by $\ind(v)$. We denote the maximum in-degree of $D$ as $\Delta_{\mathrm{in}}(D)$; similarly, we denote the maximum out-degree as $\Delta_{\mathrm{out}}(D)$. If $D$ is undirected, then $\Delta_{\mathrm{in}}(D) = \Delta_{\mathrm{out}}(D) = \Delta(D)$.

Let $q \ge 2$, we denote $[q]=\{0,1,\dots,q-1\}$. For all $x = (x_1, \dots, x_n) \in [q]^n$, we use the following shorthand notation for all $J = \{j_1, \dots, j_k\} \subseteq V$: $x_J = (x_{j_1}, \dots, x_{j_k})$. For all $x,y\in[q]^n$ we set $\Delta(x,y):=\{i\in [n]\,:\,x_i\ne y_i\}$. The \textbf{Hamming distance} between $x$ and $y$ is $\dH(x,y)=|\Delta(x,y)|$. The Hamming weight of $x \in [q]^n$ is $\wH(x) = \{i : x_i \ne 0\} = \dH(x, (0,\dots,0))$. The volume of a ball of Hamming radius $t$ in $[q]^n$ does not depend on its centre, hence we denote it as 
$$
	\VH(q,n,t) = |\{ x \in [q]^n : \wH(x) \le t \}| = \sum_{d=0}^t \binom{n}{d} (q-1)^d.
$$

Let $f:[q]^n\to[q]^n$ be a Finite Dynamical System (FDS). The image of $f$ is denoted as $\Ima(f)$. We write the FDS as $f = (f_1, \dots, f_n)$ where $f_v : [q]^n \to [q]$ is a \textbf{local function} of $f$. We also use the shorthand notation $f_J : [q]^n \to [q]^{|J|}$, $f_J = (f_{j_1}, \dots, f_{j_k})$. We associate with $f$ the digraph $G(f)$, referred to as the {\bf interaction graph} of $f$, defined by: the vertex
set is $V$; and for all $u,v\in V$, there exists an arc $(u,v)$ if and only if $f_v$ depends essentially on $x_u$, i.e. there exist $x,y\in [q]^n$ that
only differ by $x_u\neq y_u$ such that $f_v(x)\neq f_v(y)$. For a digraph $D$, we denote by $F(D,q)$ the set of FDSs $f:[q]^n\to[q]^n$ with $G(f)\subseteq D$.

If $q$ is a prime power, we shall endow $[q]$ with the finite field structure $\GF(q)$. In that case, we say that $f$ is \textbf{linear} if every local function is a linear combination of the local variables in $x$: $f_v(x) = \sum_{i=1}^n m_{i,v} x_i$. In other words, $f$ is linear if it is of the form $f(x) = xM$ for some matrix $M \in \GF(q)^{n \times n}$. We say that $f$ is \textbf{affine} if $f = xM + y$ for some matrix $M \in \GF(q)^{n \times n}$ and some vector $y \in \GF(q)^n$.

\subsection{Guessing game} \label{sec:guessing}

The \textbf{guessing number} comes from the hat game called ``guessing game'' \cite{Rii06, Rii07, GR11}, where the team wins if and only if all players guess correctly, and the aim is to maximise the number of winning configurations. More formally, for any $f$, the set of fixed points of $f$ is denoted $\fix(f) = \{x \in [q]^n : f(x) = x\}$. Then the \textbf{guessing number} of $f$ is defined by
$$
	g(f) :=\log_q |\fix(f)|.
$$ 
Then the \textbf{$q$-guessing number} of $D$ is $g(D,q) :=\max_{f\in F(D,q)} g(f)$.

The guessing number has been studied in the context of network coding solvability \cite{Rii06, Rii07, GR11}. Most importantly, 
$$
	\nu(D) \le g(D,q) \le \tau(D)
$$
for all $q \ge 2$. If $g(D,q) = \tau(D)$, we then say that $D$ is \textbf{$q$-solvable}. The guessing number of the complete graph is $g(K_n, q) = n - 1$ for all $q \ge 2$, where the solution is
$$
	f_v(x) = - \sum_{u \ne v} x_u \mod q,
$$
hence the complete graph is $q$-solvable for all $q$. Moreover, the guessing number tends to a limit for large $q$: the guessing number of $D$ is $g(D) := \sup_q g(D,q) = \lim_{q \to \infty} g(D,q)$. If $g(D) = \tau(D)$, then we say that $D$ is \textbf{asymptotically solvable}. Clearly, if $D$ is $q$-solvable for some $q$, then it is asymptotically solvable. Thirdly, we can restrict the choice of FDSs to linear ones, thus yielding the linear guessing number $g_{\linear}(D,q)$. If $g_{\linear}(D,q) = \tau(D)$, we say that $D$ is \textbf{$q$-linearly solvable}. We also denote $g_{\linear}(D) := \max_q g_{\linear}(D,q)$. We could define the analogue for affine FDSs, but it is easy to check that $g_{\affine}(D,q) = g_{\linear}(D,q)$. We shall use the subscripts $\linear$ and $\affine$ throughout this paper.

We consider two interesting families of graphs for the guessing number. Firstly, the family of odd undirected cycles $\{ C_{2k+1} : k \ge 2 \}$ satisfies $n = 2k+1$ and \cite{Rii07, CM11}
$$
	\nu(C_{2k+1}) = k < g(C_{2k+1}) = k + 1/2 < \tau(C_{2k+1}) = k+1.
$$ 
These graphs are interesting for they are not asymptotically solvable.

Secondly, the \textbf{strong product} of two digraphs $D_1$ and $D_2$, denoted as $D_1 \boxtimes D_2$  is defined as such. Its vertex set is the cartesian product $V(D_1) \times V(D_2)$, and there is an arc from $(u_1,u_2)$ to $(v_1,v_2)$ if and only if either $u_1 = v_1$ and $(u_2, v_2) \in E(D_2)$, or $u_2 = v_2$ and $(u_1, v_1) \in E(D_1)$, or $(u_1,v_1) \in E(D_1)$ and $(u_2,v_2) \in E(D_2)$. Equivalently, the adjacency matrix of the strong product is given by
\begin{equation} \label{eq:D_product} \nonumber
    A_{D_1 \boxtimes D_2} = (I_{n_1} + A_{D_1}) \otimes (I_{n_2} + A_{D_2}) - I_{n_1 n_2},
\end{equation}
where $\otimes$ denotes the Kronecker product of matrices. The graph $\vec{C}_3^2 = \vec{C}_3 \boxtimes \vec{C}_3$ is displayed in Figure \ref{fig:C32}. Then, for every value of the girth $\gamma \ge 3$, the sequence $\{ \vec{C}_\gamma^k : k \ge 1\}$ satisfies $n = \gamma^k$, $\nu(\vec{C}_\gamma^k) = \gamma^{k-1}$, and
$$
	\tau(\vec{C}_\gamma^k) = g_{\linear}(\vec{C}_\gamma^k, q) = g(\vec{C}_\gamma^k) = \gamma^k - (\gamma - 1)^k
$$
for any prime power $q$ \cite{GR11}. Hence these graphs are always linearly solvable. Moreover, these graphs can be arbitrarily relatively sparse, and yet their guessing number is relatively close to the number of vertices.

%
%

\subsection{Index coding, public entropy, and the guessing graph} \label{sec:NC_IC}

The \textbf{$q$-public entropy} of $D$ is defined as the smallest amount of information that the players in the guessing game need to always guess correctly. More formally, let $H(D,q)$ denote the set of functions $h : [q]^n \to [q]^n$ such that there exists a family of functions $\{f^{(a)} \in F(D,q) : a \in \Ima(h)\}$ such that $x \in \fix(f^{(h(x))})$ for all $x \in [q]^n$. The function $h$ helps the players all guess correctly: given the value of $a = h(x)$, the players then choose to guess according to $f^{(a)}$. Then \cite{Rii06, Rii07}
$$
	b(D,q) = \min_{h \in H(D,q)} \log_q |\Ima(h)|.
$$
We have $b(D,q) \ge n - \tau(D)$. Denoting the public entropy of $D$ as $b(D) := \inf_{q \ge 2} b(D,q)$, again we obtain that $b(D) = \lim_{q \to \infty} b(D,q)$. The guessing number and the public entropy are closely related. Firstly, we always have $g(D,q) + b(D,q) \ge n$. Secondly, the solvability problem is equivalent to the analogue for the public entropy in three ways
\begin{enumerate}
	\item \textbf{Solvability equivalence.} For any $D$ and $q$, $g(D,q) = \tau(D)$ if and only if $b(D,q) = n - \tau(D)$. This is the equivalence between index coding and network coding in \cite{Rii07, Gad13, EEL13, Maz14}.
	
	\item \textbf{Asymptotic equivalence.} For any $D$, $b(D) = n - g(D)$ \cite{GR11}.
	
	\item \textbf{Linear equivalence.} For any $D$ and $q$, $b_{\linear}(D,q) = n - g_{\linear}(D,q)$ (see \cite{GR11} for instance).
\end{enumerate}

The \textbf{guessing graph} $\G(D,q)$ is the undirected graph with vertex set $[q]^n$ and where $x,y \in [q]^n$ are adjacent if and only if there is no $f \in F(D,q)$ such that $x,y \in \fix(f)$. More concretely, $\G(D,q)$ has vertex set $[q]^n$ and edge set $E = \bigcup_{v =1}^n \{\{x,y\} : x_{\inn(v)} = y_{\inn(v)}, x_v \ne y_v\}$. The guessing graph was first introduced in \cite{BBJK06, ALSWH08}, where it was referred to as ``confusion graph''. It was then independently introduced in \cite{GR11} and extended in two different fashions in \cite{Gad13, GRR15}.

By definition, any set of fixed points of some function $f \in F(D,q)$ is an independent set of $\G(D,q)$. Conversely, any independent set of the guessing graph is fixed by some FDS in $F(D,q)$; thus $g(D,q) = \log_q \alpha(\G(D,q))$, where $\alpha$ denotes the independence number \cite{GR11}. The public entropy can be viewed as the (logarithm of the) minimum number of parts in a partition into sets of fixed points:
$$
	b(D,q) = \min \left\{ \log_q |B| : B \subseteq F(D,q), \bigcup_{f \in B} \fix(f) = [q]^n \right\},
$$
hence $b(D,q) = \log_q \chi(\G(D,q))$, where $\chi$ denotes the chromatic number.

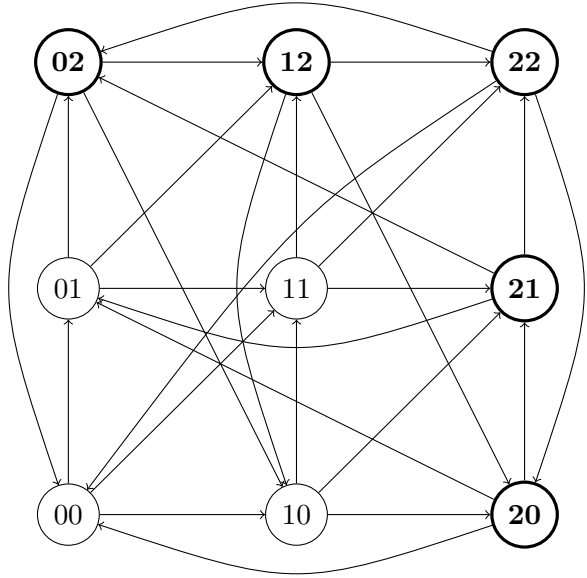
\begin{figure}
\centering
	\begin{tikzpicture}
	\tikzstyle{every node}=[draw, shape = circle];
	
	\node (00) at (0,0) {00};
	\node (01) at (0,3) {01};
	\node[very thick] (02) at (0,6) {${\bf 02}$};
	\node (10) at (3,0) {10};
	\node (11) at (3,3) {11};
	\node[very thick] (12) at (3,6) {${\bf 12}$};
	\node[very thick] (20) at (6,0) {${\bf 20}$};
	\node[very thick] (21) at (6,3) {${\bf 21}$};
	\node[very thick] (22) at (6,6) {${\bf 22}$};
	
	\draw[->] (00) -- (01);
	\draw[->] (00) -- (10);
	\draw[->] (00) -- (11);
	
	\draw[->] (01) -- (02);
	\draw[->] (01) -- (11);
	\draw[->] (01) -- (12);
	
	\draw[->] (02) .. controls(-1,3) .. (00);
	\draw[->] (02) -- (12);
	\draw[->] (02) -- (10);
	
	\draw[->] (10) -- (11);
	\draw[->] (10) -- (20);
	\draw[->] (10) -- (21);
	
	\draw[->] (11) -- (12);
	\draw[->] (11) -- (21);
	\draw[->] (11) -- (22);
	
	\draw[->] (12) .. controls(2,3) .. (10);
	\draw[->] (12) -- (22);
	\draw[->] (12) -- (20);
	
	\draw[->] (20) -- (21);
	\draw[->] (20) .. controls(3,-1) .. (00);
	\draw[->] (20) -- (01);
	
	\draw[->] (21) -- (22);
	\draw[->] (21) .. controls(3,2) .. (01);
	\draw[->] (21) -- (02);
	
	\draw[->] (22) .. controls(7,3) .. (20);
	\draw[->] (22) .. controls(3,7) .. (02);
	\draw[->] (22) .. controls(3,4) .. (00);
	
	\end{tikzpicture}
\caption{The digraph $\vec{C}_3^2$, with linear guessing number $5$. A minimum feedback vertex set is highlighted.}
\label{fig:C32}
\end{figure}

\subsection{Winkler's hat game} \label{sec:winkler}

In Winkler's hat game, the team scores a point for every correct guess. The main problem is: How many points can the team be guaranteed to score for any possible configuration of hats? In the language of FDS, let us define the \textbf{stability} of a FDS $f$ by 
$$
	s(f) := \min_{x\in[q]^n} n - \dH(x,f(x)),
$$
and the \textbf{$q$-stability} of $D$ is $s(D,q) = \max_{f \in F(D,q)} s(f)$. 

Then on a clique of size $q$, the optimal solution is to cover all possibilities of the sum of all $x_i$'s, i.e.
$$
	f_v(x) = v - \sum_{u \ne v} x_u \mod q.
$$
This guarantees exactly one correct guess (for $v = \sum_{u =1}^n x_u \mod q$) for any value of $x$. By double counting, this is the best possible. We note the similarity between the solutions for the guessing game and for Winkler's hat game on the clique. In general, packing disjoint copies of $K_q$ in the complete graph $K_n$ yields
$$
	s(K_n, q) = \left \lfloor \frac{n}{q} \right \rfloor.
$$
In particular, $s(K_n,q) = 0$ if $q > n$ and hence for any $D$ $s(D,q) = 0$ if $q > n$.

Some work has been done on $s(D,q)$, or on determining whether $s(D,q) > 0$, in which case $D$ is \textbf{$q$-stable} (it is referred to as $q$-solvable in \cite{GG15}, but that would obviously be confusing here). If $D$ is undirected, then $s(D,2) = M$, the size of a maximum matching in $D$; in general, $\nu(D) \le s(D,2) \le \tau(D)$ \cite{BHKL08}. For any $q$, there exists a $q$-stable bipartite undirected graph $D$ (first proved in \cite{BHKL08}, then refined in \cite{GG15}). Moreover, there exists a $4$-stable oriented graph \cite{GG15}. In \cite{GG15}, the authors ask whether there exists a $q$-stable oriented graph for all $q \ge 5$; we shall give an emphatic affirmative answer to that question in Theorem \ref{th:high_i}. 

Conversely, some graphs have been proved to be not $q$-stable. If $D$ is an undirected tree, then $D$ is not $3$-stable \cite{BHKL08}. The complete bipartite graph $K_{m,s}$ is not $(m+2)$-stable for any $s \ge 1$ \cite{GG15}. If $\tau(D) = 1$, then $D$ is not $3$-stable \cite{GG15}. The undirected cycle $C_n$ is $3$-stable if and only if $n=4$ or $n$ is divisible by $3$; moreover, $C_n$ is not $4$-stable for all $n$ \cite{Szc14}.

\section{The guessing and coset dimensions} \label{sec:guessing_code}

For any $f$, the \textbf{guessing code} and the \textbf{guessing dimension} of $f$ are
\begin{align*}
	C_f &:= \{ f(x) - x : x \in [q]^n \},\\
	l(f) &:= \log_q |C_f|.
\end{align*}
The $q$-\textbf{guessing dimension} of $D$ is then denoted as $l(D,q) := \min_{f \in F(D,q)} l(f)$, and also $l(D) := \inf_{q \ge 2} l(D,q)$. The guessing dimension of a graph is closely related to the public entropy $b(D,q)$.

\begin{lemma} \label{lem:l_covering}
For any $D$ and $q$,
$$
	l(D,q) = \min \left\{ \log_q |S| :  S \subseteq [q]^n, \, \exists f \in F(D,q) : \bigcup_{a \in S} \fix(f - a) = [q]^n \right\},
$$
whence $l(D,q) \ge b(D,q) \ge n - \tau(D)$.
\end{lemma}

\begin{proof}
We have 
\begin{align*}
	\bigcup_{a \in S} \fix(f - a) = [q]^n &\Leftrightarrow \forall x \in [q]^n \, \exists a \in S : x \in \fix(f-a)\\
	&\Leftrightarrow  \forall x \in [q]^n \, \exists a \in S : f(x) - x = a\\
	&\Leftrightarrow C_f \subseteq S,
\end{align*}
and the equation follows.
\end{proof}

We now define the \textbf{coset dimension} of $f$ by
$$
	c(f) := \min \left\{ \log_q |S| :  S \subseteq [q]^n, \, \exists f \in F(D,q) : \bigcup_{a \in S} ( \fix(f) - a) = [q]^n \right\},
$$
and the \textbf{$q$-coset dimension} of $D$ is $c(D,q) := \min_{f \in F(D,q)} c(f)$, and also $c(D) := \inf_{q \ge 2} c(D,q)$.

\begin{lemma} \label{lem:c_covering}
For any $D$ and $q$, $c(D,q) \ge b(D,q) \ge n - \tau(D)$.
\end{lemma}

\begin{proof}
For any $f \in F(D,q)$ and any $a \in [q]^n$, define $g \in F(D,q)$ by $g(x) = f(x+a) - a$. Then $\fix(g) = \fix(f) - a$, and by definition of $c(D,q)$ and $b(D,q)$, $c(D,q) \ge b(D,q)$.
\end{proof}

We strengthen the triple equivalence between network coding and index coding by replacing the public entropy by the coset dimension.

\begin{theorem}[Equivalences] \label{th:g+l=n}
We have three kinds of equivalences.
\begin{enumerate}
	\item \textbf{Solvability equivalence.} For any $D$ and $q$, the following are equivalent.
	\begin{enumerate}
		\item $g(D,q) = \tau(D)$.
		\item $b(D,q) = n - \tau(D)$.
		\item $c(D,q) = n - \tau(D)$.
	\end{enumerate}
	
	\item \textbf{Asymptotic equivalence.} For any $D$,
	$$
		c(D) = \lim_{q \to \infty} c(D,q) = b(D) = n - g(D).
	$$
	
	\item \textbf{Linear equivalence.} For any $D$ and $q$,
	\begin{align*}
		l_{\affine}(D,q) &= l_{\linear}(D,q) = c_{\affine}(D,q) = c_{\linear}(D,q)\\
		c_{\linear}(D,q) &= b_{\linear}(D,q) = n - g_{\linear}(D,q).
	\end{align*}
\end{enumerate}
\end{theorem}

\begin{proof}
1. Solvability equivalence. The equivalence of the first two properties is reviewed in Section \ref{sec:background}. Suppose that $c(D,q) = n - \tau(D)$, then by Lemma \ref{lem:c_covering}, $b(D,q) = n - \tau(D)$ and hence $g(D,q) = \tau(D)$. Conversely, let $g(D,q) = \tau(D)$ and let $\fix(f) = \{z^i : i=1, \dots, q^{\tau(D)}\}$ be an independent set of $\G(D,q)$ of size $q^{\tau(D)}$. Let $I$ be a minimum feedback vertex set of $D$. 

Firstly, $z^i_I \ne z^j_I$ for all $i \ne j$. Indeed, by the proof of \cite[Proposition 5]{GR11}, the set $\Delta(z^i, z^j)$ must contain a cycle of $D$, hence it intersects $I$. Secondly, for any $a_J \in [q]^{n - \tau(D)}$, let $a = (0_I, a_J)$. Then for any $i,j$,  $\Delta(z^i -a, z^j -a) = \Delta(z^i, z^j)$ and hence $z^i -a \not\sim z^j -a$. Thus denoting $S := \{a : a_J \in [q]^{n - \tau(D)}\}$, we have $\bigcup_{a \in S} (\fix(f) - a) = [q]^n$ and by Lemma \ref{lem:l_covering}, $c(D,q) = n - \tau(D)$.\\
~\\
2. Asymptotic equivalence. We prove that
$$
	n - g(D,q) \le c(D,q) \le n - g(D,q) + \log_q(1 + q \ln n).
$$
The lower bound follows from Lemma \ref{lem:c_covering}, while the upper bound comes from the proof of \cite[Proposition 3.12]{Bab94} and the fact that the guessing graph is vertex-transitive.\\
~\\
3. Linear equivalence. It is easy to see that $l_{\linear}(D,q) = l_{\affine}(D,q)$ and that $c_{\affine}(D,q) = c_{\linear}(D,q)$ (by translation).  Then let us denote $C^{-1}_f(z) := \{x \in [q]^n : f(x) - x = z\}$. If $f$ is linear, then $|C^{-1}_f(z)| = \frac{q^n}{|C_f|}$ for all $z$.  By translation, we easily obtain
\begin{align*}
	g_{\linear}(D,q) &= \log_q \max_{f \in F(D,q), f \text{ linear}} \max_{z \in [q]^n} |C^{-1}_f(z)|\\
	&= \log_q \max_{f \in F(D,q), f \text{ linear}} \frac{q^n}{|C_f|}\\
	&= n - l_{\linear}(D,q).
\end{align*}
Finally, it is clear that if $f$ is linear, then $c(f) = n - g(f)$, hence $c_{\linear}(D,q) = n - g_{\linear}(D,q)$.
\end{proof}

\section{Stability and instability} \label{sec:instability}

For an FDS $f$, we define the \textbf{instability} of $f$ as
$$
	i(f) :=\min_{x\in[q]^n} \dH(x,f(x)).
$$
For any digraph $D$, the \textbf{$q$-instability} of $D$ is $i(D,q) :=\max_{f\in F(D,q)} i(f)$, and again the instability of $D$ is $i(D) = \max_{q \ge 2} i(D,q)$. 

\subsection{General properties} \label{sec:instability_general}

\begin{proposition} \label{prop:i_increasing}
For every digraph $D$ we have $i(D,2)=s(D,2)$. Moreover, for every digraph $D$ and $q\ge 2$ we have 
\begin{align*}
	i(D,q) &\le i(D,q+1),\\
	s(D,q) &\ge s(D,q+1).
\end{align*}
\end{proposition}

\begin{proof}
For every $f\in F(D,2)$, we have $i(f)=s(\neg f)$ and $\neg f\in F(D,2)$, hence $i(D,2)=s(D,2)$ and $i(D,2) = s(D,2)$. 

For every $x\in [q+1]^n$ let $x'\in[q]^n$ be defined by $x'_i=\min(x_i,q - 1)$ for all $i\in [n]$. Let $f:[q]^n\to [q]^n$ and let $f':[q+1]^n\to [q+1]^n$ be defined by $f'(x)=f(x')$ for all $x\in[q+1]^n$. It is easy to see that $G(f')=G(f)$. Furthermore $\dH(x,f'(x)) \ge \dH(x',f(x'))$ for all $x\in [q+1]^n$, and thus $i(f')\geq i(f)$. The proof for the stability number is similar.
\end{proof}

\begin{proposition} \label{prop:bounds_i}
For every digraph $D$ and $q\geq 2$ we have 
\[
	\nu(D) \le i(D,q) \le \tau(D).
\]
\end{proposition}

\begin{proof}
For a cycle $C$ we have $i(C,2)=1$, by the following function:
$$
	f_i(x) = \begin{cases}
		\neg x_n &\text{if } i = 1\\
		x_{i-1} &\text{if } 2 \le i \le n.
	\end{cases}
$$
Indeed, if $x$ is a fixed point, we have $x_n + 1 = x_1 = x_2 = \dots = x_n$. By Proposition \ref{prop:i_increasing}, we obtain $i_{\affine}(C,q) \ge 1$ for all $q \ge 2$ and thus the lower bound is obtained by packing negative cycles. For the upper bound, let $f\in F(D,q)$ and let $I$ be a feedback vertex set of $D$ with of minimum size. Since $D-I$ has no cycle we deduce that for every $x\in [q]^I$ there exists at least one $x'\in[q]^n$ such that $x'_I = x$ and $\Delta(x',f(x'))\subseteq I$. Then $i(D,q)\le \dH(x',f(x'))\leq |I|=\tau(D)$. 
\end{proof}

\begin{proposition} \label{prop:i(Kn)}
For every $n\geq q\geq 2$, 
\[
	i(K_n,q)=n-\left\lceil \frac{n}{q}\right\rceil.
\]
\end{proposition}

\begin{proof}
Let $n=kq+r$ where $k$ and $r$ are integers and $0\leq r<q$. The classical solution of Winkler's hat game shows that $i(K_p,p)=p-1$ for every $p\geq 2$. Thus if $p\leq n$ then $i(K_p,n)\geq p-1$. Therefore, by taking the union of $k$ disjoint copies of $K_q$ and one residual $K_r$ we obtained a spanning subgraph $H$ of $K_n$ such that if $r > 0$,
\[
i(K_n,q)\geq i(H,q)\geq k(q-1)+r-1=n-\left\lceil\frac{n}{q}\right\rceil, 
\]
and the same end result holds for $r=0$.

Conversely, let $f\in F(K_n,q)$ with $i(f)=i(K_n,q)$. By double counting, again we obtain  
\[
	q^n i(K_n,q) = q^ni(f) \le n(q-1)q^{n-1}
\] 
showing that $i(K_n,q)\le \left\lfloor n - \frac{n}{q}\right\rfloor$.
\end{proof}

For any $f$, we have $s(f) + i(f) \ge n$ by definition. However, we obtain the result in the opposite direction for digraphs.

\begin{corollary} \label{cor:s+i}
For any $D$ and $q$, $s(D,q) + i(D,q) \le n$.
\end{corollary}

\subsection{Suprema of stability and instability} \label{sec:i=tau}

We know that the stability $s(D,q) = 0$ for $q$ large enough; we shall also prove in Theorem \ref{th:i=tau} below that $i(D,q) = \tau(D)$ for $q$ large enough. Therefore, we also investigate how fast these asymptotic bounds are reached.

For any $D$, let $E'(D)$ be the set of chordless cycles of $D$, let $L(D)$ be the undirected graph on $E'(D)$ such that two chordless cycles are adjacent if they meet in at least one vertex of $D$, and let $\chi'(D)$ be the chromatic number of $L(D)$. In particular,  if $D$ is undirected, $E'(D)$ is the edge set of $D$, $L(D)$ is the line graph of $D$, and $\chi'(D)$ is the chromatic index of $D$. According to Vizing's theorem, $\chi'(D) \in \{\Delta(D), \Delta(D) + 1 \}$ if $D$ is undirected \cite{BM08}.

\begin{theorem} \label{th:i=tau}
For any digraph $D$, 
$$
	i(D) = \lim_{q \to \infty} i(D,q) = \tau(D).
$$
Moreover, we have $i(D,q) = \tau(D)$ if $q = 2^{\chi'(D)}$ or if $q = 2^{\Delta(D)}$ and $D$ is undirected.
\end{theorem}

\begin{proof}
If $D$ is acyclic, then $i(D,q) = 0 = \tau(D)$ for all $Q$. Otherwise, by definition, we can partition the set of chordless cycles into $\chi = \chi'(D)$ parts 
$$
	\{C^{1,1}, \dots, C^{1,p_1}\}, \dots, \{C^{\chi,1}, \dots, C^{\chi,p_\chi}\}
$$
such that $C^{\alpha, i}, C^{\alpha,j}$ are disjoint for all $1 \le \alpha \le \chi$ and $1 \le i < j \le p_\alpha$. Denote each chordless cycle by
$$
	C^{\alpha, i} = (u_1^{\alpha, i}, \dots, u_{l_{\alpha, i}}^{\alpha, i}).
$$
Let $q = 2^\chi$: we then view $x_v \in [q]$ as $x_v = (x[v,1], \dots, x[v,\chi])$. Then consider the function $f \in F(D,q)$ where
\begin{align*}
	f_v(x) &= (f[v,1] (x), \dots, f[v,\chi] (x))\\
	f[v, \alpha] (x) &= \begin{cases}
		\neg x[ u_{l_{\alpha, i}}^{\alpha,i}, \alpha ] &\text{if } v = u_1^{\alpha, i}\\
		x[ u_{k-1}^{\alpha, i}, \alpha ] &\text{if } v = u_k^{\alpha, i}, k > 1\\
		0 &\text{otherwise}.
	\end{cases}
\end{align*}
We claim that $J$ is acyclic. Indeed, if $J$ contains the chordless cycle $C^{\alpha, i}$, then 
$$
	x[ u_1^{\alpha, i}, \alpha ] = x[ u_2^{\alpha, i}, \alpha ] = \dots = x[ u_{l_\alpha}^{\alpha, i}, \alpha ] = \neg x[ u_1^{\alpha, i}, \alpha ].
$$
Since $J$ is acyclic, its complement is a feedback vertex set, of cardinality at least $\tau(D)$.

If $D$ is undirected, for all $v \in [n]$ let $\inn(v) = \{u_1, \dots, u_{\ind(v)}\}$ sorted in increasing order. Then $P$ be the $n \times \Delta(D)$ matrix such that 
$$
	P(v,d) = \begin{cases}
	u_d &\text{if } 1 \le d \le \ind(v)\\
	0 &\text{if } d > \ind(v).
	\end{cases}
$$
Also, let $Q$ be the $n \times n$ matrix such that for all $v$ and $d \le \ind(v)$,
$$
	Q(P(v,d),v) = d,
$$
and all other entries are zero. Therefore,
$$
	P[P(v,d), Q(v, P(v,d))] = v, \quad Q\{ P(v,d), P[ P(v,d), Q(v, P(v,d)) ] \} = d.
$$

Let $q = 2^{\Delta(D)}$: we then view $x_v \in [q]$ as $x_v = (x[v,1], \dots, x[v, \Delta(D) ])$. Consider the function $f \in F(D,q)$ where
\begin{align*}
	f_v(x) &= (f[v,1] (x), \dots, f[v, \Delta(D) ] (x))\\
	f[v,d] (x) &= \begin{cases}
		 x[ P(v,d), Q(v, P(v,d) ) ] &\text{if } 1 \le P(v,d) < v\\
		 \neg x[ P(v,d), Q(v, P(v,d) ) ] &\text{if } v < P(v,d) \le n\\
		0 &\text{if } P(v,d) = 0.
	\end{cases}
\end{align*}
Again, for every $x$, denote the set of coordinates $j$ such that $f_j(x) = x_j$ as $J$. We claim that $J$ is independent. Indeed, if $J$ contains the edge $\{v, P(v,d)\}$ with $P(v,d) < v$ then
\begin{align*}
	x[v, d] &= f[v, d](x)\\
	&= x[ P(v,d), Q(v, P(v,d) ) ]\\
	&= f[ P(v,d), Q(v, P(v,d)) ] (x)\\
	&= \neg x[ P[ P(v,d), Q(v, P(v,d)) ], Q\{ P(v,d), P[ P(v,d), Q(v, P(v,d)) ] \} ]\\
	&= \neg x[v, d].
\end{align*}
\end{proof}


\begin{corollary} \label{cor:i_C5}
For any $k \ge 2$, 
$$
	k = i(C_{2k+1},2) < g(C_{2k+1},2) < k + 1/2 = g(C_{2k+1}) < k+1 = i(C_{2k+1}). 
$$
In particular, there exists a digraph $D$ such that for all $q$ large enough, $g(D,q) < i(D,q) = \tau(D)$.
\end{corollary}

\begin{theorem} \label{th:stability_upper_bound}
For any $D$ and $q$ and $m \ge 1$, let $Q(m) = 2 + \sum_{a=1}^m a^a$, then
\begin{align*}
	s(D,q) &\le \frac{\tau(D)}{\lfloor (q-1)^{1/\tau(D)} \rfloor},\\
	s(D,Q(m)) &\le \tau(D) - m,
\end{align*}
and in particular $s(D, Q(\tau(D))) = 0$.
\end{theorem}

\begin{proof}
Let $\tau = \tau(D)$, $I$ be a minimum feedback vertex set of $D$ and $J = V \setminus I = \{j_1, \dots, j_{n-\tau} \}$ in topological sort. Let $f \in F(D,q)$ with $s(f) = s(D,q)$.

\begin{claim*} \label{claim:1}
For any set $X \subseteq [q]^\tau$ such that $|X| < q$, there exists $y^{(X)} \in [q]^{n-\tau}$ such that 
$$
	\dH(y^{(X)}, f_J(x, y^{(X)}) ) = n- \tau
$$
for all $x \in X$.
\end{claim*}

\begin{subproof}
Recursively define $y = y^{(X)} \in [q]^{n-\tau}$ such that
\begin{align*}
	y_{j_1} &\notin f_{j_1}(X)\\
	y_{j_2} &\notin f_{j_2}(X, y_{j_1})\\
	&\vdots\\
	y_{j_{n-\tau}} &\notin f_{j_{n-\tau}}(X,y_{j_1}, \dots, y_{j_{n-\tau-1}}).
\end{align*}
\end{subproof}

Firstly, let $p = \lfloor (q-1)^{1/\tau} \rfloor$, so that $q \ge p^\tau + 1$ and consider the set $X = [p]^\tau \subseteq [q]^\tau$. Then we claim that there exists $x \in X$ such that 
$$
	\dH(x, f_I(x,y^{(X)})) \ge \tau - s(D[I],p).
$$
Indeed, otherwise the function $g \in F(D[I], p)$ defined as $g(x) = \max\{ p-1, f_I( x, y^{(X)} ) \}$ has stability greater than $s(D[I],p)$, which is a contradiction. Thus, $s(D,q) \le s(D[I],p) \le s(K_\tau, p) \le \frac{\tau}{p}$.

Secondly, let $q = Q(m)$ and consider the sets $X_l \subseteq [l+1]^\tau \subseteq [q]^\tau$ defined recursively by $X_1 = \{(0,0,\dots,0) , (1,0,\dots,0) \}$ and for all $2 \le l \le m$,
\begin{align*}
	A_l &= \{x \in [q]^\tau: x_{1,\dots,l} \in [l]^\tau , x_{l+1} = \dots = x_\tau = 0\},\\
	B_l &= \{x \in [q]^\tau: (x_{1,\dots,l-1}, 0, \dots, 0) \in X_{l-1}, x_l = l, x_{l+1} = \dots = x_\tau = 0\},\\
	X_l &= A_l \cup B_l.
\end{align*}
Then $|X_l| = Q(l) - 1$ and we claim that for all $1 \le l \le m$, there exists $x \in X_l$ such that 
$$
	\dH(x_{1,\dots,l}, f_{1,\dots,l}(x, y^{(X_m)}) ) = l.
$$
We prove it by induction on $l$; the case $l = 1$ is clear. Suppose it holds for $l-1$ but not for $l$. Note that in $X_l$, the value of $(x_{l+1}, \dots, x_\tau)$ is fixed, so we write $f_i(x_{1,\dots,l})$ for any $1 \le i \le l$. Then in $A_l$, there is always one player from $1$ to $l$ who guesses correctly, hence for every value of $x_{1,\dots,l-1} \in [l]^\tau$, there exists $x_l$ such that $x_{1,\dots,l}$ is guessed correctly by player $l$. In other words, $f_l(x_{1,\dots,l-1}) \in [l]$ for any $x_{1,\dots,l-1} \in [l]^{l-1}$. Now in $B_l$, the players $1$ to $l-1$ cannot always guess correctly, by induction hypothesis. Thus, there exists $z \in [l]^{l-1}$ such that $f_l(z) = l$, which is a contradiction. Thus, $s(D,q) \le \tau - m$.
\end{proof}

\begin{corollary}
If $\tau(D) = 2$, then $s(D, 7) = 0$. Also, for any $D$, $s(D,3) \le \tau(D) - 1$, i.e. the $\tau(D)$ upper bound on the stability can  only be reached for $q=2$.
\end{corollary}

%

\subsection{Relation with the guessing number} \label{sec:i_v_g}

We can relate the guessing number of a digraph with its stability and instability.

\begin{theorem} \label{th:i_v_g}
For every digraph $D$ we have 
\begin{align*}
	g(D,q) &\geq \tau(D) - \log_q \left[ q^\tau(D) - \VH(q, \tau(D), i(D,q)-1) \right],\\
	g(D,q) &\ge \tau(D) - \log_q \VH(q, \tau(D), \tau(D) - s(D,q)).
\end{align*}
\end{theorem}

\begin{proof}
We begin with an important property of acyclic sets.

\begin{claim*} \label{claim:2}
Let $I$ be a feedback vertex set of $D$ and $J = V \setminus I$. Then for any $f \in F(D,q)$, $x \in [q]^I$, $a \in [q]^J$, there exists $x^{(a)} \in [q]^n$ such that $x^{(a)}_I = x$ and $f(x^{(a)}) - x^{(a)} = a$.
\end{claim*}

\begin{subproof}
We sort $J$ in topological order $J = \{j_1, \dots, j_k\}$ and we construct $x^{(a)}_J$ recursively. We have
\begin{align*}
	x^{(a)}_{j_1} &= f_{j_1}(x_I) - a_{j_1}\\
	x^{(a)}_{j_2} &= f_{j_2}(x_I, x^{(a)}_{j_1}) - a_{j_2}\\
	&\vdots\\
	x^{(a)}_{j_k} &= f_{j_1}(x_I, x^{(a)}_{j_1}, \dots, x^{(a)}_{j_{k -1}}) - a_{j_k}.
\end{align*}
\end{subproof}

Let $I$ be a feedback vertex set of $D$ of size $\tau = \tau(D)$ and $J = V \setminus I$. Let $f \in F(D,q)$ with maximal instability,  and 
$$
	\mathcal{Y} = \{ y \in [q]^\tau : \wH(y) \ge i(f) \}.
$$
By the Claim, for every $x\in [q]^I$ there exists a unique point $x'\in [q]^n$ such that $x'_I = x$ and $y = f(x') - x'$ satisfies $y_J  = (0, \dots, 0)$, hence $y_I \in \mathcal{Y}$. The function $x\mapsto \delta(x) = y_I$ is thus a function from $[q]^I$ to $\mathcal{Y}$, hence there exists $a \in \mathcal{Y}$ such that 
\[
	|\delta^{-1}(a)| \ge \frac{|[q]^I|}{|\mathcal{Y}|} = \frac{q^\tau}{q^\tau - \VH(q, \tau, i(D,q)-1)}. 
\]
Consider then the FDS $f' \in F(D,q)$ defined as
$$
	f'_v(x) = \begin{cases}
		f_v(x) - a_v &\text{if } v \in I\\
		f_v(x) &\text{if } v \in J.
	\end{cases}
$$
For every $x\in \delta^{-1}(a)$, $x'$ is a fixed point of $f'$. Since $x\mapsto x'$ is an injection, $g(f') \ge \log_q |\delta^{-1}(a)|$, which combined with the above, proves the result.

The proof for the stability is similar. Let $f \in F(D,q)$ with maximal instability, $I$ and $J$ as above, and
$$
	\mathcal{Z} = \{ z \in [q]^\tau : \wH(z) \le \tau - s(f) \}.
$$
By the Claim, for every $x\in [q]^I$ there exists a unique point $x''\in [q]^n$ such that $x''_I = x$ and $z = f(x'') - x''$ satisfies $z_J  = (1, \dots, 1)$, hence $z_I \in \mathcal{Z}$. The function $x\mapsto \epsilon(x) = z_I$ is thus a function from $[q]^I$ to $\mathcal{Z}$, and hence there exists $b \in \mathcal{Z}$ such that 
\[
	|\epsilon^{-1}(b)| \ge \frac{|[q]^I|}{|\mathcal{Z}|} = \frac{q^\tau}{\VH(q, \tau, \tau - s(D,q))}. 
\]
Consider then $f' \in F(D,q)$ defined as
$$
	f''_v(x) = \begin{cases}
		f_v(x) - b_v &\text{if } v \in I\\
		f_v(x) - 1 &\text{if } v \in J.
	\end{cases}
$$
For every $x\in \epsilon^{-1}(b)$, $x''$ is a fixed point of $f''$. Since $x\mapsto x''$ is an injection, $g(f'') \ge \log_q |\epsilon^{-1}(b)|$, which combined with the above, proves the result.
\end{proof}

\begin{corollary} \label{cor:i_v_g}
If $i(D,2) = s(D,2) = \tau(D)$, then $g(D,2) = \tau(D)$.
\end{corollary}

The implication does not hold for all $q$, as seen from Corollary \ref{cor:i_C5}.

%

\section{Linear and affine (in)stability} \label{sec:linear}

\subsection{Digraphs with high affine stability and instability} \label{sec:high_i}

\begin{theorem} \label{th:high_i}
For any $\gamma \ge 3$, and any $\epsilon > 0$, there exists $D$ with girth $\gamma$ such that
\begin{align*}
	i_{\affine}(D, 2) &> (1 - \epsilon) \frac{n}{2},\\
	i_{\affine}(D) = \tau(D) &> (1 - \epsilon) n,\\
	s_{\affine}(D, \lceil \epsilon^{-1} \rceil) &> 0.	
\end{align*}
\end{theorem}

The strategy is to recast the problem in terms of metric properties of codes, and then to use our results on the guessing code. A $q$-ary code $C$ of length $n$ is a subset of $[q]^n$. The \textbf{remoteness} and the \textbf{covering radius} of $C$ are respectively defined as \cite{MS77, CHLL97, CG12}
\begin{align*}
	\rem(C) &:= \min_{y \in [q]^n} \max_{c \in C} \dH(c,y)\\
	\cov(C) &:= \max_{y \in [q]^n} \min_{c \in C} \dH(c,y).
\end{align*}

\begin{lemma} \label{lem:i=remoteness}
For any $D$ and any $q$,
\begin{align*}
	i(D,q) &= \max_{f \in F(D,q)} \cov(C_f),\\
	s(D,q) &= \max_{f \in F(D,q)} n - \rem(C_f).
\end{align*}
The same results hold in the affine case.
\end{lemma}

\begin{proof}
We only prove the first equality, the second one being very similarly proven.

Firstly, we can express $f$ as $f = \phi + f(0)$, where $\phi(0) = 0$ and $G(f) = G(\phi)$ (where $0$ denotes the all-zero vector of length $n$). Thus, for any graph $D$, the set $F(D,q)$ can be partitioned into classes of the form $\{\phi - y: y \in [q]^n\}$. Then
\begin{align*}
	i(D,q) &= \max_{\phi \in F(D,q), \phi(0) = 0} \max_{y \in [q]^n} \min_{x \in [q]^n} \dH(x, \phi(x) - y)\\
	&= \max_{\phi \in F(D,q),  \phi(0) = 0} \max_{y \in [q]^n} \min_{x \in [q]^n} \dH(\phi(x) - x, y)\\
	&= \max_{\phi \in F(D,q), \phi(0) = 0} \max_{y \in [q]^n} \min_{c \in C_\phi} \dH(c, y)\\
	&= \max_{\phi \in F(D,q), \phi(0) = 0} \cov(C_\phi).
\end{align*}
Finally, since $C_f = C_{\phi + f(0)} = C_\phi + f(0)$, we have $\cov(C_f) = \cov(C_\phi)$, which concludes the proof.
\end{proof}

If $C_f$ is small, it has a high covering radius, yielding high instability; it also has low remoteness, thus yielding high stability as well.

\begin{lemma} \label{lem:i_v_l}
For any $D$ and $q$,
\begin{align*}
	\log_q \VH(q, n, i(D,q)) &\ge n - l(D,q),\\
	\log_q (q^n - \VH(q, n, n - s(D,q) - 1)) &\ge n - l(D,q).
\end{align*}
The same results hold in the affine case.
\end{lemma}

\begin{proof}
The sphere-covering bound \cite{CHLL97} states that for any code $C$ of covering radius $\rho$,
$|C| \VH(q,n,\rho) \ge q^n.$ Moreover, if $C$ has remoteness $r$, then \cite{CG12} $|C| (q^n - \VH(q,n,r - 1) ) \ge q^n.$
The results then follow from applying these bounds to $C = C_f$ of cardinality $q^{l(D,q)}$.
\end{proof}

\begin{proof}[Proof of Theorem \ref{th:high_i}]
Let $H_2(p) = - p \log_2 p - (1-p) \log_2 (1-p)$ denote the binary entropy for $0 \le p \le 1$. Denote $q = 2^{\lceil - \log_2 \epsilon \rceil} > \lceil \epsilon^{-1} \rceil$ and $D = \vec{C}_\gamma^k$, where $k$ is chosen such that
$$
	\tau(D) > n \max \left\{ H_2 \left( \frac{1 - \epsilon}{2}  \right), 1 - \epsilon, \log_q (q-1) \right\}.
$$
By Theorem \ref{th:g+l=n}, $l_{\affine}(D,2) = l_{\affine}(D,q) = \tau(D)$.

Firstly, by \cite[Chapter 10, Corollary 9]{MS77}, we obtain
$$
	H_2 \left( \frac{i(D,2)}{n} \right) \ge \log_2 \VH(2,n,i(D,2)) \ge \tau(D) > H_2 \left( \frac{1- \epsilon}{2}  \right),
$$
and hence $i(D,2) > (1 - \epsilon) \frac{n}{2}$. Secondly, $i(D) = \tau(D) = n(1-\epsilon)$. Thirdly, suppose $s(D,q) = 0$, then
$$
	\log_q \left(q^n - \VH(q,n,n - s(D,q) - 1) \right) = \log_q (q-1)^n < \tau(D) = n - l(D,q),
$$
which violates Lemma \ref{lem:i_v_l}.
\end{proof}

\subsection{Additional properties}

Firstly, note that $i_{\linear}(D,q) = 0$, since any linear FDS fixes the all-zero vector. Also, from our preliminary results on the stability and instability, we have $i_{\affine}(D,2) = s_{\affine}(D,2)$ and $i_{\affine}(D,q) \ge \nu(D)$.

We now prove that the stability of a linear FDS is severely limited by its interaction graph.

\begin{proposition} \label{stability_linear}
For any linear FDS $f$,
$$
	s(f) \le n - \Delta_{\mathrm{out}}(G(f)) - 1.
$$
\end{proposition}

\begin{proof}
We have $f(x) - x = xM$, where $M$ has support $I_n + A_{G(f)}$. Therefore,
$$
	s(f) = n - \max_{x\in[q]^n} \dH(x,f(x)) = n - \max_{c \in C_f} \wH(c).
$$
Since the rows of $M$ are codewords of $C_f$, the maximum weight is at least $\Delta_{\mathrm{out}}(G(f)) + 1$. 
\end{proof}

This bound is trivially achieved if $\Delta_{\mathrm{out}}(G(f)) = n-1$, in which case $s(f) = 0$. More interestingly, it is also achieved for the graphs constructed from the cyclic simplex codes ($n = 2^r-1$, $\Delta = 2^{r-1} - 1$, $s_{\linear}(D,2) = 2^{r-1} - 1$). See \cite{GR11} for how to construct a digraph from a cyclic code. In particular, for $r=3$ we obtain another example of an oriented graph $D$ with $i(D,2) = \lfloor n/2 \rfloor$, namely the Paley tournament on seven vertices displayed on Figure \ref{fig:P7}, the first example being the directed triangle.

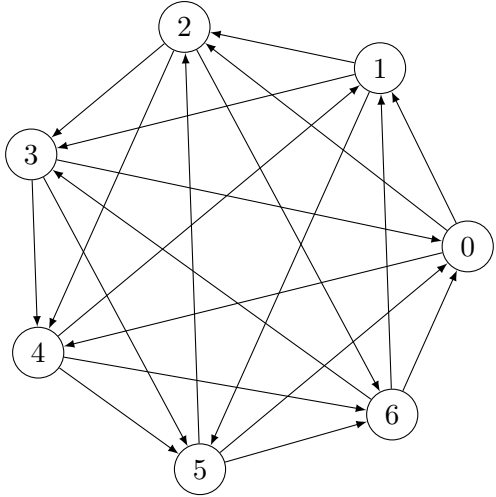
\begin{figure}
\centering
\begin{tikzpicture}
\tikzstyle{every node}=[draw, shape=circle];

\foreach \x in {0,...,6}{
	\node (\x) at (52*\x:3) {\x};
}

	\draw[-latex] (0) -- (1);
	\draw[-latex] (0) -- (2);
	\draw[-latex] (0) -- (4);
	
	\draw[-latex] (1) -- (2);
	\draw[-latex] (1) -- (3);
	\draw[-latex] (1) -- (5);

	\draw[-latex] (2) -- (3);
	\draw[-latex] (2) -- (4);
	\draw[-latex] (2) -- (6);
	
	\draw[-latex] (3) -- (4);
	\draw[-latex] (3) -- (5);
	\draw[-latex] (3) -- (0);

	\draw[-latex] (4) -- (5);
	\draw[-latex] (4) -- (6);
	\draw[-latex] (4) -- (1);
	
	\draw[-latex] (5) -- (6);
	\draw[-latex] (5) -- (0);
	\draw[-latex] (5) -- (2);

	\draw[-latex] (6) -- (0);
	\draw[-latex] (6) -- (1);
	\draw[-latex] (6) -- (3);
\end{tikzpicture}
\caption{The Paley tournament on seven vertices with binary instability $3$.}
\label{fig:P7}
\end{figure}

Although we do not know whether the affine instability always reaches the feedback upper bound, we can prove that it always exceeds the linear guessing number.

\begin{theorem} \label{th:i_affine}
For any $D$, $i_{\affine}(D) := \sup_{q \ge 2} i_{\affine}(D,q) \ge g_{\linear}(D)$.
\end{theorem}

\begin{proof}
Due to \cite[Theorem 4.3]{WCR09}, it is easy to check that $g_{\linear}(D,p^m) \ge g_{\linear}(D,p)$ for any prime power $p$ and any integer $m \ge 1$. Therefore, there exists $q$ large enough so that $n \log_q 2 < \epsilon$ and $n - l_{\affine}(D,q) = g_{\linear}(D,q) = g_{\linear}(D)$. Then let $i := i_{\affine}(D,q)$; we have
$\VH(q,n,i) \le 2^n q^i$ and hence by Lemma \ref{lem:i_v_l} and Theorem \ref{th:g+l=n},
$$
	i + \epsilon > \log_q \VH(q,n,i) \ge n - l_{\affine}(D,q) = g_{\linear}(D).
$$
\end{proof}

\section{Acknowledgment}

The author would like to thank Alonso Castillo-Ramirez and Adrien Richard for stimulating discussions.


\providecommand{\bysame}{\leavevmode\hbox to3em{\hrulefill}\thinspace}
\providecommand{\MR}{\relax\ifhmode\unskip\space\fi MR }
\providecommand{\MRhref}[2]{%
  \href{http://www.ams.org/mathscinet-getitem?mr=#1}{#2}
}
\providecommand{\href}[2]{#2}

\end{document}